\documentclass[12pt,a4paper,openbib]{article}
\usepackage[latin1]{inputenc}
\usepackage{amsmath}
\usepackage{amsfonts}
\usepackage{amssymb}
\usepackage{makeidx}
\usepackage{graphics} 
\usepackage{graphicx} 
\usepackage{pstricks,pst-node} 
\usepackage{tikz} 
\usepackage{stmaryrd}
\usepackage{bm}
\usepackage{diagbox}
\usepackage{amsthm}
\usepackage{authblk}
\renewcommand{\textbf}[1]{\begingroup\bfseries\mathversion{bold}#1\endgroup}

\begin{document}

\newtheorem{defi}{Definition}
\newtheorem{thm}[defi]{Theorem}
\newtheorem{clm}[defi]{Claim}
\newtheorem{lmm}[defi]{Lemma}

\date{}
\author[1]{\small Louis Mathieu}
\author[2]{\small Mehdi Mhalla}
\affil[1]{\footnotesize Univ. Grenoble Alpes}
\affil[2]{\footnotesize Univ. Grenoble Alpes, CNRS, Grenoble INP, LIG, F-38000 Grenoble France}
\title{Separating pseudo-telepathy games and two-local theories}

\maketitle
   \begin{abstract}
We give an $\dfrac{1}{54}$ separation between 5-party pseudo-telepathy games and two-local theories. We define the notion of strategy in a k-local theory for a game, and extend the method of \cite{test}. We also study variation of the game to minimize the classical winning probability.
   \end{abstract}

\section{Introduction}
Nonlocal correlations are one of the most fascinating aspects of quantum mechanics predictions. It offers possibilities that surpass local classical models. Such correlations between multiparties can be obtained in different ways : \begin{itemize}
\item Physically : The different parties can share entanglement through a quantum device.
\item Theoretically : Sharing non-signaling resources such as 'nonlocal box' or other tools, enables parties to study simulations of nonlocal correlations.
\end{itemize}
The term non-signaling means that, given a set of parties $(1, \ldots,n)$,  who receive $(x_1,\ldots,x_n)$ and answer $(a_1,\ldots,a_n)$, the marginal probability distribution $P(a_1,\ldots,a_{i-1},a_{i+1},\ldots,a_n | x_1,\ldots,x_n)$ is independant of $x_i$ \cite{NS}.

In 1964, John S. Bell gave inequalities that can be used to prove the existence of quantum correlations by repeating specific tests and compare observed values with adapted theoretical values. Construction of these inequalities required a lot of work to be properly tested and accepted. 
Nevertheless, operations and measurement on a shared quantum state between multiple parties give non-signaling correlations, and for example maximum algebraic violation of the CHSH inequality can be reached by a non-signalling box.
Thus, we need some experiments that can highlight dissimilarities between these two kinds of tools. 

Nonlocal games 	are very efficient to study this part of quantum information theory. Their goal is to reveal nonlocal correlations  that cannot statistically appear under classical theories. A nonlocal game involves two or more players, spatially separated to ensure they don't communicate. They receive questions  and have to produce an answer satisfying certain conditions to win. Many works have been performed to understand them well. To separate quantum theories from classical theories, they could design "self-tests" that allows users to verify operations of quantum devices without trusting it \cite{self-testing,self}. 

The nonlocality of quantum correlation is less powerful than general non-signaling tools. However, some games have been designed to separate robustly quantum from non-signaling theories with restriction of locality. For example, pseudo-telepathy games \cite{broadbent} are games that can be won perfectly by players sharing quantum resources but not if they share only local randomness \cite{magic,self,pseudo}. In \cite{contex}, the notion of "contextuality width" is introduced to show the impact of locality of shared non-signaling information \cite{pseudo,pironio,contex}. That is what interests us, we want to separate quantum theories from two-local non-signaling ones. In \cite{test}, these two correlations are separated from $5.1 \%$, but it relies on test that can't be won perfectly by quantum players. 

Nonlocal games can be studied with different approaches. First, we can see it as a combinatorial game, where players have questions, they produce answer and we check if their answer satisfy the rule of the game. The game is determined by a set containing pairs of questions and answers. If they produce an answer associated to their question belonging to the losing set of the game, they lose. Then, the goal of the players is to collaborate, elaborating the best strategy possible to try to never give answers in the losing set. These are called Multipartite collaborative games \textbf{MCG} \cite{contex}.

On the other hand, we can define a game by associating a protocol where an external verifier (referee) picks questions from a known probability distribution for the players. After, the referee checks if players answered correctly to their question. This allows us to study shrewdly the probability to win to the game by analysing question by question the probabilities. This is the angle of study in \cite{self,test,XOR} and in this work. 
Game questions can have different structures. First, this can be simple questions, chosen to imply probabilities to win and ease the analysis, such as the $C_5$ game defined in \cite{self} or even games in \cite{test,three}.
Secondly, questions can be numerous but easily generated such as in \cite{three} or in \cite{pseudo, contex} where questions are generated by graph states.

When the losing predicates depend on the xor of the player's answers, the game belongs to the family of XOR games. It has been studied with linear algebraic tools in \cite{XOR}.  However two cases are possible, either all the answers are taken in account for all questions, either specific players are chosen to be "involved" for each question.

The object of this study was to find a separation between quantum theories and two local non-signaling theories, like in \cite{test}, we tried to find it with pseudo telepathy games such as the $C_5$ game \cite{pseudo}.

\section{Preliminaries}

\begin{flushleft}
 \textbf{Game notations.} Let G be a n-player game. Let $(x_1,\ldots,x_n) \in \lbrace 0,1 \rbrace^n$ be their inputs and $(a_1,\ldots,a_n) \in \lbrace 0,1 \rbrace^n$ their outputs. We call a question $Q$ an element of $\lbrace 0,1 \rbrace^n$ as inputs for players. For each question, an external verifier will check a predicate to validate their answers. If the answer $a_v$ of a player $v$ to a question $Q$ can modify the validity of the predicate, we say that this player is \textit{involved} in question $Q$. 
 
 We introduced the following set to characterize games : $\mathcal{G} := \lbrace (p,Q,I,s) \rbrace$ where $Q$ is a question, $p$ is the probability for $Q$ to be asked, $I$ is the set of involved players in this question and $s$ is the value associated to the following win condition : $$\sum\limits_{i \in I} a_i=s.$$ Then a game is totally determinate by this set.
 \end{flushleft} 
 
To study properly non-signaling correlations with locality restriction, \cite{test} introduced the notion of k-party multiround resource :
	\begin{defi}
	A k-party multi-round non-signaling resource is a conditional probability distribution $P(A^{(1)},...,A^{(k)} \; | X^{(1)},...,X^{(k)})$ satisfying, for all $j_1,...,j_k \; \geq 0$ $$ P(A^{(1)}_{1:j_1},...,A^{(k)}_{1:j_k} \; | X^{(1)},...,X^{(k)})=P(A^{(1)}_{1:j_1},...,A^{(k)}_{1:j_k} \; | X^{(1)}_{1:j_1},...,X^{(k)}_{1:j_k})$$
	Where $A^{(k)}$ and $X^{(k)}$ are bit vectors, corresponding respectively to outputs and inputs players can have with this resource, and $A^{(1)}_{1:j_k}$ corresponds to the $j_k$ first coordinates of the vector.
	\end{defi}
This tool enables to consider only one resource per set of players instead of unlimited one-round resources between these players. 
Then, to extend the result of \cite{test}, we will use the following definition and lemma :
	\begin{defi}\textbf{\cite{test}}
	A question $Q^{'}$ is \textbf{v-compatible} with $Q$ if, either v's input differs from $Q^{'}$ to $Q$, or for any other player u involved in $Q$, either u's input is the same as in $Q$, either u is not involved in $Q^{'}$.
	\end{defi}
	
	\begin{lmm}\textbf{\cite{test}}
	Consider a unique nonlocal game. Let S be a strategy, for classical players using non-signaling resources, that wins with probability at least $1-\varepsilon$ on all questions. Fix a player $v$ and a non-signaling resource $R$ involving $v$, and assume that there exists a question $Q$ such that the verifier's acceptance predicate depends only on the responses of $v$ and players $U$ not involved in $R$. Assume further that for any other question $Q^{'}$ in which $v$'s input is the same as in $Q$, for any player $u \in U$ either $u$'s input is the same as in $Q$ or the verifier's acceptance predicate does not depend on $u$'s response.
	
	Then there exists a strategy $S^{'}$ that wins with probability at leat $1-3\varepsilon$ on all questions, and that is the same as $S$ except for $v$'s behavior on its input in question $Q$; on this input, $v$ ignores the resource $R$, and in fact we have $P(S^{'} \; loses|Q) \leq P(S \; loses|Q)$ and $P(S^{'} \; loses|Q^{'}) + 2p(S \; loses|Q)$ for any $Q^{'}$ in which $v$'s input is the same as in $Q$.
	\end{lmm}

\section{Strategy}	
	
	We will now explain a core notion for the rest of the study, the notion of strategy. We call strategy for players in a specific game, a protocol that describes the behaviour of players in this game. The nature of the resources players can access will define the nature of the strategy. 
	For these results, we consider that players have access to arbitrary k-party multi-round non-signaling resources and cannot communicate. Note that, even though players can generate nonlocal correlations, accessing quantum resources or other non-signaling local resources, they still have a classical behaviour. The notion of nonlocality just comes from the nature of their resources.

	In the following of the paper, we denote by $\bigcup \limits_X \mathcal{R}_{v,X}$ the resources a player can access (can be indexed by an integer $k$) and ${\mathcal{R}_k}_{|v}=a_k^{|v}$ the bit he receives from it. Then, a player $v$ receiving $x_v^0$ as first input for the game, can reiterate the following actions as many times as he wants : \begin{itemize}
	\item Send a bit in one of the resources he can access
	\item Read a bit that the resource outputs
	\item Make unlimited computation with the bits he received and local variables. 
	\end{itemize}
	If we look at the definition of k-party multi-round non-signaling resources, we observe that a resource take as input bit strings from any player that have access to it, and they receive bit strings as outputs, distributed respecting the non-signaling conditions and causality conditions. 
	We can then resume a strategy for a player with this graph : 
	
\includegraphics[scale=0.4]{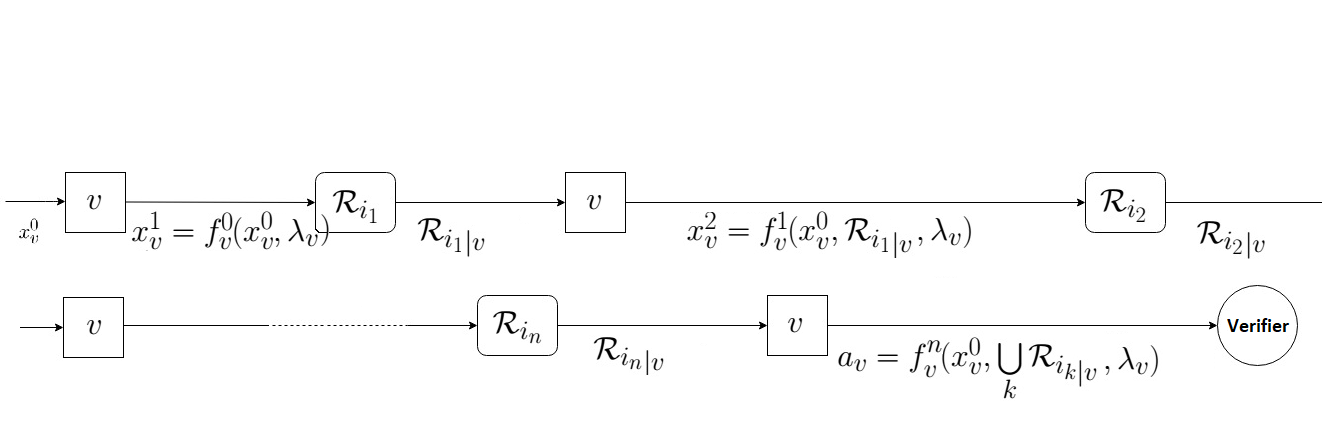} 
	
	We see in the previous graph that the choice of a strategy can be resumed as chosing the resources used, shared and local randomness $\lambda_v$, and functions $ f_v^1,\ldots,f_v^n$, $g_1,\ldots,g_n$ such that $i_k=g_k(x_v^0,\bigcup\limits_{j \leq k} a_j^{|v},\lambda_v).$ In this model, the length of the protocol of a player v and the index of the resources he can access are not fixed beforehand. We can parametrize the length by functions $(s_k)_{k \in \mathbb{N}}$ where the protocol stops at step $q$ if $s_q(x_v^0,\bigcup\limits_{k \leq q} a_k^{|v},\lambda_v)=0$. We will consider that the length of the protocol is always finite for players.
	
	We present now the Five-player cycle graph game which is an example of pseudo-telepathy games using graph states. This is the game used by Pironio and Barrett to separate quantum from two-local theory \cite{barrett}.
	\section{Five-player cycle graph game}
	
	The Five-player cycle graph game is a protocol with players $1,2,3,4$ and $5$. The input and output of player i are denoted respectively $x_i$ and $a_i$, where the indices are taken modulo $5$.
	For this game, we consider the following questions :
		\begin{itemize}
			\item Questions $Q_i=0_{i-2}0_{i-1}1_i0_{i+1}0_{i+2}$ 
			\item Question $Q_a=1_{i-2}1_{i-1}1_i1_{i+1}1_{i+2}$ 
		\end{itemize}
		Then the characterising set associated to this game is $$\mathcal{G}_{C_5}:=\Big\lbrace \Big(\dfrac{1}{6},Q_a,\lbrace 1,2,3,4,5 \rbrace,1\Big), \;\Big(\dfrac{1}{6},Q_i,\lbrace i-1,i,i+1 \rbrace,0\Big)\Big\rbrace$$
	All the questions are chosen with equal probability.\\
	
	This game is "unique" \cite{test}, i.e. if we take a player $i$ involved in the question $Q$, and if we fix responses from the other involved players, there exists only one response for $i$ that makes the verifier accept.\\
	It is a pseudo-telepathy game, it means that there exists a quantum strategy which allows the player to win with probability one \cite{self,pseudo}. We will try to find a gap between this probability, and the probability to win with strategies using non-signaling resources. This will be in two steps. The first two claims find the relation between probabilities to win when all the players use their non-signaling resources and when a set of players don't. The second is to prove that, at a certain point, all the players can't do better than players with only classical resources, even if they still use some non-signaling resources. The reasoning follows the one in \cite{test} but, their hypothesis on the game were not the same, thus we needed to modify and extend their result. 
	So, we prove the following theorem :
	
	\begin{thm}
 In the Five-player cycle graph game, classical players sharing arbitrary two-local non-signaling resources can win with probability at most $1-\dfrac{1}{54}$.
	\end{thm}

	\vspace{0,5cm}
	\begin{proof}
	
	\begin{clm}
	Let S be a strategy, for classical players using non-signaling resources, that wins with probability at least $1-\varepsilon$ on all questions.
	
	Then there exists a strategy $S_{\lbrace 1,5,2 \rbrace}$ that wins with probability at least $1-\varepsilon$ to $Q_a$, $1-5\varepsilon$ to $Q_2$, $Q_4$ and $Q_5$, $1-13\varepsilon$ to $Q_3$ and $1-17\varepsilon$ to $Q_1$, where on input 0, players 1, 5 and 2 don't use their non-signaling resources. 
	\end{clm}
	 
	 \begin{proof}Let S be a strategy that wins with at least $1-\varepsilon$ on all questions, for classical players using two-party non-signaling resources.  \label{notation}
	
	We now consider the player $1$, involved in  non-signaling resources $\mathcal{R}_{1,i} \;, \forall i \in \lbrace 2,3,4,5 \rbrace$. Let \bm{$Q=Q_2$}. We will  prove that there exists a strategy $S^{'}$ for these players that win with a least $1- 3\varepsilon$ on all questions.
	
	First, we need to introduce the notion of randomness of a non-signaling resource. The goal of this analysis is, given a game, a strategy, and a question, to calcul the probability to win for the players. We saw in the introduction that each strategy induced a probability distribution for the outputs of the players. Thus, given a question Q and a strategy S, we have that $$P(S \; wins |Q):=\sum\limits_{(a_1,\ldots,a_n)\; winning} p(a_1,\ldots,a_n | x_1,\ldots,x_n).$$ Here, we use a tool to have a factorisation of this probability, to have it in the following way : $$P(S \; wins |Q):=\sum\limits_{(a_1,\ldots,a_n)\; winning} p(a_1 |Q)p(a_2|Q,a_1)\ldots p(a_n|Q,a_1,\ldots,a_{n-1}).$$
	In the analysis, to calcul this probability to win, we can parametrize the randomness of the non-signaling resources taking a specific order \cite{test}. 
    To do this, the analysis consists in chosing an order in players. Then, given a player $v$ and bits he receives from resources he can access $\bigcup\limits_k {\mathcal{R}_{i_k}}_{|v}$, we look at the outputs he received from it. In the general case, when you observe first the outputs of his resources that come to him, they follow a certain distribution. Then, in the analysis, we parametrize it by a random variable, that would describe the outputs of this resource in this setting, and we sum over all the possible realisation of this randomness. We do it for all the resources for a player, and the realisation of this random variable will determine the final output of this player to the game. Here, we parametrize the randomness of resources of player $\lbrace 2,3 \rbrace$ first, then $1$, and finally $\lbrace 4,5 \rbrace$.
	\begin{itemize}
		\item Let $r_1=\bigcup\limits_{x \neq 2, y\neq 3}r_{2,x_{|2}}\cup r_{3,y_{|3}}$ denote the randomness for $\lbrace 2,3 \rbrace$. It fixes the answers of the players $2$ and $3$ to $Q$. Let $G(r_{1})$ denote the unique answer for $1$ for which the verifier accepts.
		\item Let $r_2=r_{1,2_{|1}}\cup r_{1,3_{|1}}$ and $r_3=r_{1,4_{|1}}\cup r_{1,5_{|1}}$ denote the remaining randomness needed to determine the outputs at $1$ of any resources involving $1$. Thus the answer of player $1$ is a deterministic function of $(x_1^0,r_1,r_2,r_3)$ and local variables $\lambda_1$
	\end{itemize}
	
	By assumption on strategy $S$ we have the following inequality with the above parametrization : 
	$$1- \varepsilon \leq \sum \limits_{r_1,r_2,r_3}p(r_1)p(r_2)p(r_3) \cdot \chi_{G(r_1)=f_1^n(x_1^0,r_1,r_2,r_3,\lambda_1)}.$$
	
	In this analysis, we see that the randomness of resources $\mathcal{R}_{1,4}$ and $\mathcal{R}_{1,5}$ is observed last. In particular, a way to block the use of this resource for player 1 would be to give him a fixed random variable $r_3^{*}$, which distribution suits in the case of question $Q$ and this order of analysis. 
	
	In this setting, the probability to win the game given question $Q$ and with player $1$ not using resources with players $\lbrace 4,5 \rbrace$ satisfies  $$1- \varepsilon \leq \sum \limits_{r_1,r_2,r_3^{*}}p(r_1)p(r_2)p(r_3^{*})\cdot \chi_{G(r_1)=f_1^n(x_1^0,r_1,r_2,r_3^*,\lambda_1)}.$$
	
	We define the following strategy $S^{'}$ that is the same as $S$ except when $x_1=0$. In this case, the player $1$ doesn't use any resource with $\lbrace 4,5 \rbrace$ (\bm{$\mathcal{R}^1=\mathcal{R}_{1,4}  \cup \mathcal{R}_{1,5} $}) and simulates its behaviour with the fixed $r_3^*$. It means that, in the strategy described as above, everytime the player $1$ would have used resources with $4$ or $5$, he will simulate the outputs of these boxes with this random distribution instead. 
	
	We will analyse question by question the probability to win for players using this strategy : \begin{itemize}
		\item For $Q_2$ : The inequality above holds for strategy $S^{'}$ by definition, thus $$P(S^{'} \; loses | Q_2)\leq\varepsilon.$$
		\item For $Q_1$ and $Q_a$ : $x_1=1$ thus the player $1$ will have the same behaviour as in strategy $S$, such as every other players, therefore $$P(S^{'} \; loses | Q_1)=P(S^{'} \; loses | Q_6)\leq \varepsilon.$$
		\item For $Q_5$, we can apply the Lemma 6 \cite{test} : The game we play is a unique nonlocal game. We have S a strategy that wins with probability at least $1- \varepsilon$. Here, we took player $1$, and the non-signaling resource \bm{$\mathcal{R}^1=\mathcal{R}_{1,4}  \cup \mathcal{R}_{1,5} $} involving $1$. The question $Q_2$ is such that the verifier's acceptance predicate depends only on the responses of $1$ and $\lbrace 2,3 \rbrace$ not involved in $\mathcal{R}^1$. 
Here, the question $Q_5$ is 1-compatible with $Q_2$  : indeed, the input of player 1 is the same, and even if inputs of player $2$ and $3$ are not the same, they are not involved in $Q_5$. The lemma tells us that for $Q^{'}$, a 1-compatible question with $Q_2$, we have the following inequality  $$P(S^{'} \; loses|Q^{'}) \leq P(S \; loses|Q^{'}) + 2 P(S \; loses|Q_2) \leq 3\varepsilon$$ thus, in this case we obtain that $$P(S^{'} \; loses|Q_5) \leq P(S \; loses|Q_5) + 2 P(S \; loses|Q_2) \leq 3\varepsilon.$$
		
		\item For $Q_3$ (A similar argument can be used for $Q_4$), we can't apply the lemma used for $Q_5$, so we need to fin an other argument. However, we notice that in $Q_3$, the verifier's acceptance predicate doesn't depend on $1$'s response, so intuitively, we would say that changing the behaviour for player $1$ only won't change the probability to win on this question.
		
		\hspace{0,5cm} Indeed, if we analyze differently strategy $S$, parameterizing the randomness for the non-signaling resources according to the players $\lbrace 2,3 \rbrace$ going first, then $4$, then $\lbrace 1,5 \rbrace$. Let $r_4=\bigcup\limits_{x \neq 4} r_{4,x_{|4}}$ the remaining randomness needed to determine the outputs of $4$ to $Q_3$, its response will be a deterministic function $f_4^n(x_4^0,r_1,r_4,\lambda_4)$. And by assumption on strategy $S$ we have for question $Q_3$ : 
		$$ 1- \varepsilon \leq \sum \limits_{r_1,r_4}p(r_1) p(r_4)\cdot \chi_{G(r_1)=f_4^n(x_4^0,r_1,r_4,\lambda_4)}.$$
		The right-hand side of this equation is exactly the probability that $S'$ wins for $Q_3$.
		Similarly for $Q_4$, we obtain $$P(S' \; loses | Q_3)=P(S' \; loses | Q_4)\leq \varepsilon.$$
		
	\end{itemize} 
	 
	\vspace{1cm}
	
	We will now apply a similar reasoning but considering player $1$, and \bm{$Q=Q_5$}. And we define a strategy $S^{''}$ which is the same as $S^{'}$ except if $x_1=0$ ; in this case, the player $1$ doesn't use resource \bm{$\mathcal{R}^2=\mathcal{R}_{1,2}\bigcup \mathcal{R}_{1,3}$}.
	
	We obtain the following inequalities : \begin{itemize}
		\item \bm{$P(S^{''} \; loses | Q_5) \leq P(S^{'} \; loses | Q_5)\leq 3\varepsilon.$}
		\item $P(S^{''} \; loses | Q_1)=P(S^{''} \; loses | Q_a)\leq \varepsilon.$
		\item $P(S^{''} \; loses | Q_3)=P(S^{''} \; loses | Q_4)\leq \varepsilon.$
		\item \bm{$P(S^{''} \; loses | Q_2) \leq P(S^{'} \; loses | Q_2) + 2 P(S^{'} \; loses | Q_5)\leq 7 \varepsilon$}
	\end{itemize}
	
 	Finally, we obtained a strategy $S^{''}$ which wins with at least probability $1-7\varepsilon$ on question $Q_2$, $1-3\varepsilon$ on question $Q_5$ and $1- \varepsilon$ on all other questions.
 	
 	But these probabilities depend on the order we removed different resources to player $1$. We could define a symetric strategy $S^{'''}$ in which we first removed $\mathcal{R}^2=\mathcal{R}_{1,2}\bigcup \mathcal{R}_{1,3}$ and then $\mathcal{R}^1=\mathcal{R}_{1,4}  \cup \mathcal{R}_{1,5} $. This strategy would have resulted to inverted probabilities for $Q_2$ and $Q_5$, but similar ones on other questions.
 		
 	It seems natural now to consider the average strategy $S_{\lbrace 1 \rbrace}$ which consists in chosing a uniformly random strategy between $S^{''}$ and $S^{'''}$. This strategy consists in giving to the player 1 a random variable at the beginning, uniformly distributed in $\llbracket 0,1 \rrbracket$ and he will use the fixed randomness to choose his behaviour according to the value of this random variable. 
 	
 	\vspace{0,5cm}
 	To resume, we obtain the following probability table : 
 	
 	\vspace{0,3cm}

	 \begin{center}
	 	\begin{tabular}{|c|c|c|c|c|c|c|}
	 		\hline 
	 		\multicolumn{7}{|c|}{P($S$ loses $\mid$ $Q$) $\leq$} \\ 
	 		\hline 
	 		\rule[-1ex]{0pt}{2.5ex}\backslashbox{Strategy $S$}{\vrule width 0pt height 1.25em Question $Q$} & $Q_1$& $Q_2$ & $Q_3$ & $Q_4$ & $Q_5$ & $Q_a$ \\ 
	 		\hline 
	 		\rule[-1ex]{0pt}{2.5ex}	$S^{''}$& $\varepsilon $&  $7\varepsilon$& $\varepsilon $& $ \varepsilon$&  $3\varepsilon$& $\varepsilon $ \\ 
	 		\hline 
	 		\rule[-1ex]{0pt}{2.5ex}$S^{'''}$& $\varepsilon $&  $3\varepsilon$& $\varepsilon $& $ \varepsilon$&  $7\varepsilon$& $\varepsilon $  \\ 
	 		\hline 
	 		\rule[-1ex]{0pt}{4ex}	$S_{\lbrace 1 \rbrace}$& $\varepsilon $&  $5\varepsilon$& $\varepsilon $& $ \varepsilon$&  $5\varepsilon$& $\varepsilon $ \\ 
	 		\hline 
	 	\end{tabular} 
	 \end{center}

The strategy obtained $S_{\lbrace 1 \rbrace}$ is the same as $S$ except when the player $1$ has input $x_1=0$, he acts like a classical player. And, we obtained that the strategy $S_{\lbrace 1 \rbrace}$ loses with probability at most $5 \varepsilon$ on questions $Q_2$ and $Q_5$. Thus, we will continue this process, defining a final strategy in which on input $0$, some players don't use their resources. For this new game, we will obtain a bound for the probability to lose, which will allow us to have a bound for the initial strategy $S$.

If we consider now the player $5$ and the question $Q=Q_1$, we will apply exactly the same reasoning, and find a strategy $S_{\lbrace 1 \rbrace}^{'}$ that is the same as $S_{\lbrace 1 \rbrace}$ except when $x_5=0$. In this case, the player 5 doesn't use any resource with $\lbrace 3,4 \rbrace$.
We obtain the following inequalities : \begin{itemize}
		\item \bm{$P(S_{\lbrace 1 \rbrace}^{'} \; loses | Q_1) \leq P(S_{\lbrace 1 \rbrace} \; loses | Q_1)\leq \varepsilon.$}
		\item \bm{$P(S_{\lbrace 1 \rbrace}^{'} \; loses | Q_4) \leq P(S_{\lbrace 1 \rbrace} \; loses | Q_4) + 2 P(S_{\lbrace 1 \rbrace} \; loses | Q_1)\leq 3\varepsilon$}
		\item $P(S_{\lbrace 1 \rbrace}^{'} \; loses | Q_2)=P(S_{\lbrace 1 \rbrace}^{'} \; loses | Q_5)\leq 5\varepsilon.$
		\item $P(S_{\lbrace 1 \rbrace}^{'} \; loses | Q_3)=P(S_{\lbrace 1 \rbrace}^{'} \; loses | Q_3)\leq \varepsilon.$

	\end{itemize}
Doing the same with player 5 and question $Q_4$, we obtain the strategy $S_{\lbrace 1 \rbrace}^{''}$ that verifies these inequalities : 
\begin{itemize}
		\item \bm{$P(S_{\lbrace 1 \rbrace}^{''} \; loses | Q_4) \leq P(S_{\lbrace 1 \rbrace}^{'} \; loses | Q_4)\leq 3\varepsilon.$}
		\item \bm{$P(S_{\lbrace 1 \rbrace}^{''} \; loses | Q_1) \leq P(S_{\lbrace 1 \rbrace}^{'} \; loses | Q_1) + 2 P(S_{\lbrace 1 \rbrace}^{'} \; loses | Q_4)\leq 7\varepsilon$}
		\item $P(S_{\lbrace 1 \rbrace}^{''} \; loses | Q_2)=P(S_{\lbrace 1 \rbrace}^{'} \; loses | Q_5)\leq 5\varepsilon.$
		\item $P(S_{\lbrace 1 \rbrace}^{''} \; loses | Q_3)=P(S_{\lbrace 1 \rbrace}^{'} \; loses | Q_3)\leq \varepsilon.$
		
	\end{itemize}

Once again, if we take the average strategy, we obtain the following table of probability, where $S_{\lbrace 1,5 \rbrace}$ is the same strategy as $S$ except when $x_1=0$, player $1$ doesn't use any of his resources, same for player $5$.

\begin{center}
	 	\begin{tabular}{|c|c|c|c|c|c|c|}
	 		\hline 
	 		\multicolumn{7}{|c|}{P($S$ loses $\mid$ $Q$) $\leq$} \\ 
	 		\hline 
	 		\rule[-1ex]{0pt}{2.5ex}\backslashbox{Strategy $S$}{\vrule width 0pt height 1.25em Question $Q$} & $Q_1$& $Q_2$ & $Q_3$ & $Q_4$ & $Q_5$ & $Q_a$ \\ 
	 		\hline 
	 		\rule[-1ex]{0pt}{2.5ex}	$S_{\lbrace 1 \rbrace}^{''}$& $7\varepsilon $&  $5\varepsilon$& $\varepsilon $& $ 3\varepsilon$&  $5\varepsilon$& $\varepsilon $ \\ 
	 		\hline 
	 		\rule[-1ex]{0pt}{2.5ex}$S_{\lbrace 1 \rbrace}^{'''}$& $3\varepsilon $&  $5\varepsilon$& $\varepsilon $& $ 7\varepsilon$&  $5\varepsilon$& $\varepsilon $  \\ 
	 		\hline 
	 		\rule[-1ex]{0pt}{4ex}	$S_{\lbrace 1,5 \rbrace}$& $5\varepsilon $&  $5\varepsilon$& $\varepsilon $& $ 5\varepsilon$&  $5\varepsilon$& $\varepsilon $ \\ 
	 		\hline 
	 	\end{tabular} 
	 \end{center}

At the end, we obtain this final table, of strategy $S_{\lbrace 1,5,2 \rbrace}$ in which players $\lbrace 1,2,5 \rbrace$ don't use their resources on input 0. \begin{center}
	 	\begin{tabular}{|c|c|c|c|c|c|c|}
	 		\hline 
	 		\multicolumn{7}{|c|}{P($S$ loses $\mid$ $Q$) $\leq$} \\ 
	 		\hline 
	 		\rule[-1ex]{0pt}{2.5ex}\backslashbox{Strategy $S$}{\vrule width 0pt height 1.25em Question $Q$} & $Q_1$& $Q_2$ & $Q_3$ & $Q_4$ & $Q_5$ & $Q_a$ \\ 
	 		\hline 
	 		\rule[-1ex]{0pt}{2.5ex}	$S_{\lbrace 1,5 \rbrace}^{''}$& $27\varepsilon $&  $5\varepsilon$& $11\varepsilon $& $5\varepsilon$&  $5\varepsilon$& $\varepsilon $ \\ 
	 		\hline 
	 		\rule[-1ex]{0pt}{2.5ex}$S_{\lbrace 1,5 \rbrace}^{'''}$& $7\varepsilon $&  $5\varepsilon$& $15\varepsilon $& $ 5\varepsilon$&  $5\varepsilon$& $\varepsilon $  \\ 
	 		\hline 
	 		\rule[-1ex]{0pt}{3ex}	$S_{\lbrace 1,5,2 \rbrace}$& $17\varepsilon $&  $5\varepsilon$& $13\varepsilon $& $ 5\varepsilon$&  $5\varepsilon$& $\varepsilon $ \\ 
	 		\hline 
	 	\end{tabular} 
	 \end{center}
\end{proof}

\begin{clm}
Let S be a strategy, for classical players using non-signaling resources, that wins with probability at least $1-\varepsilon$ on all questions.

There exists a strategy $\overline{S}$ that wins with probability at least $1-9\varepsilon$ on all questions, that is the same as S except that players share randomness that determine $i$, such that players $i$, $i+1$ and $i-1$ don't use their non-signaling resources on input $0$.
\end{clm}
\begin{proof}
The strategy $S_{\lbrace 1,5,2 \rbrace}$ obtained in previous claim defines new rules for the game. We said that players $1,5$ and $2$ couldn't use their non-signaling resources when their input was $0$. 
We could have obtained an other game strategy, with different probabilities to lose defining that $3$ other players couldn't use their non-signaling resources on input $0$. 

Thus, we can define the final strategy $\overline{S}$ in which players share random variable distributed such that a player $i$ is chosen randomly and then, these players $i$, $i+1$ and $i-1$ don't use their non-signaling resources on input $0$. There are $5 \times 2$ possibilities to remove one player and its two neighbors. By symetry on questions and players, we obtain the following table of probability for the strategy $\overline{S}$ :

\begin{center}

\begin{tabular}{|c|c|c|c|c|c|c|}
	 		\hline 
	 		\multicolumn{7}{|c|}{\rule[-1ex]{0pt}{3.5ex}P($\overline{S}$ loses $\mid$ $Q$) $\leq$} \\ 
	 		\hline 
	 		\rule[-1ex]{0pt}{2.5ex}\backslashbox{Strategy $S$}{\vrule width 0pt height 1.25em Question $Q$} & $Q_1$& $Q_2$ & $Q_3$ & $Q_4$ & $Q_5$ & $Q_a$ \\ 
	 		\hline 
	 		\rule[-1ex]{0pt}{4ex}	$\overline{S}$&  $9\varepsilon $&  $9\varepsilon$& $9\varepsilon $& $9\varepsilon$&  $9\varepsilon$& $\varepsilon $ \\ 
	 		\hline
	 	\end{tabular}
	 \end{center}
\end{proof}
For the next claim, we just define the notions of non-signaling values and classical values \cite{XOR}. The non-signaling value is the best reachable probability to win for players sharing non-signaling resources. The classical value is the same for players sharing classical correlations. 

\begin{clm}
In the five-player cycle graph game when $3$ players, representing one vertex and its neighbors, ignore their non-signaling resources on input $0$, the two-local non-signaling value is $\dfrac{5}{6}$, the same as the classical value.
\end{clm}

\begin{proof}
Assuming that $3$ players representing one vertex and its neighbors don't use their resources on input 0, we have the following scenario : \\\\
The two remaining players $i,j$ are neighbors and we will see the possible questions they receive. Bits in parenthesis represent the case when the player is not involved : 
\begin{center} \begin{tabular}{cc}
player i & player j\\
1 & 1 \\
0&(0) \\
1&0\\
0&1\\
(0)&0\\
(0)&(0) \\
\end{tabular}
\end{center}
By the non-signaling property, the two players remaining can't have any information about the other players's input, so they can't know, except concerning player j for i, or i for j, if they will use their non-signaling resources to answer the question. 
 
First, let's study what happens when the player $i$ uses his resources with player $j$, but he doesn't : we denote $I_i$ and $I_j$ their two vectors of inputs, and $(\mathcal{R}_{\lbrace i,j \rbrace})_{|i}$, $(\mathcal{R}_{\lbrace i,j \rbrace})_{|j}$ their two vectors of outputs from the resource. The final output $a_j$ won't depend on $(\mathcal{R}_{\lbrace i,j \rbrace})_{|j}$ and the player $i$ will use $a_i$. We have, by the non-signaling property : $p( (\mathcal{R}_{\lbrace i,j \rbrace})_{|i} \; | I_i I_j) = p( (\mathcal{R}_{\lbrace i,j \rbrace})_{|i} \; | I_i).$
So, the resource of player $i$ shared with $j$, can be seen as a local random variable for $i$ whose final output depend partially from this resource.

Let's consider, without loss of generality, that the two remaining players are players $2$ and $3$.
For this analysis, we will factorize the resource as : $$p(a_1 a_2 a_3 a_4 a_5 \; | x_1 x_2x_3x_4x_5)=p(a_1a_4a_5|x_1x_4x_5) \times p(a_2a_3 \; | x_2 x_3 x_1x_4x_5 a_1a_4a_5).$$

So, if we look at the behaviour of player $2$ and $3$, who can use their non-signaling resources whatever their inputs are, we get by the non-signaling property again : $$p (a_2 a_3 \; | x_2 x_3 \left[ x_k, k \neq 2,3 \right ]) = p(a_2 a_3 \; | x_2 x_3 000).$$
We have that $$(a_1 a_4 a_5) = f_1^n\Big( x_1,\bigcup\limits_{x \neq 1}(\mathcal{R}_{\lbrace 1,x \rbrace})_{|1},\lambda_1\Big)f_4^n\Big(x_4,\bigcup\limits_{x \neq 4}(\mathcal{R}_{\lbrace 4,x \rbrace})_{|4},\lambda_4\Big)f_5^n\Big(x_5,\bigcup\limits_{x \neq 5}(\mathcal{R}_{\lbrace 5,x \rbrace})_{|5},\lambda_5 \Big)$$ and for any fixed $(x_1,x_4,x_5, a_1,a_4,a_5)$, $$(a_2 a_3) = f_2^n\Big(x_2, \bigcup\limits_{x \neq 2}(\mathcal{R}_{\lbrace 2,x \rbrace})_{|2},\lambda_2 \Big)f_3^n\Big(\bigcup\limits_{x \neq 3}(\mathcal{R}_{\lbrace 3,x \rbrace})_{|3},\lambda_3 \Big).$$
As, on input $0$, the other players don't use any non-signaling resources, the behaviour of player $2$ and $3$ can be seen as a function of $x_{1,...,5}$ and $(\mathcal{R}_{\lbrace 2,3 \rbrace})_{|2},(\mathcal{R}_{\lbrace 2,3 \rbrace})_{|3} $. Thus, in this game, behaviour of player $2$ and $3$ can be simulated with a combination of local random variables, classical strategies, and the use of two-party non-signaling resources between them (which is equivalent to the use  of PR-boxes).
The conditional outputs's distribution of the two remaining players 2 and 3 is a convex combination of classical distribution and PR-box distribution. Thus we need to calculate the probability to lose to the game for all the extremal cases.

\begin{itemize}
\item The first case is when players $2$ and $3$ act classically, i.e, $$p(a_2a_3 \; | x_2x_3)=p(a_2 \; |x_2)p(a_3 \; |x_3).$$Knowing that on input $0$, players $\lbrace 1,4,5 \rbrace$ don't use their NS-resources, and, by the non-signaling property, we still have this equation : $$p \; (a_i \; | x_i  \left[ x_k, k \neq j \right ]) = p \; (a_i \; | x_i  00) \; \; \forall i \in \lbrace 1,4,5 \rbrace.$$
Thus, we can consider that all of these players have classical behaviour. And if we fix the strategy of players $1$, $4$ and $5$ to classical behaviour, we will have a new game : \\
\begin{center}
\begin{tabular}{ccc}
player 2 & player 3 & \^{S}\\
1&1 & 1+$a_1(1)$+$a_4(1)$+$a_5(1)$\\
0&(0)& $a_1(1)$+$a_5(0)$\\
1&0&$a_1(0)$\\
0&1&$a_4(0)$\\
(0)&0&$a_4(1)$+$a_5(0)$\\
(0)&(0) & $\ast$\\
\end{tabular}
\end{center}
And for this specific game, the classical value is $\dfrac{5}{6}$.
\item Now, if the distribution of the output $p(a_2a_3 \; |x_2x_3)$ is a PR-Box distribution, we notice that, in two questions, only one of the two players is involved, consequently, whatever result they wanted to output, they will be wrong in half of the cases. These two questions happen each with probability $\dfrac{1}{6}$, thus with this strategy, players lose with at least $\dfrac{1}{6}$. 
\end{itemize}

\end{proof}

Whatever the set of 3 players (representing one vertex and its neighbors) we chose, if they don't use their non-signaling resources on input 0, the five players won't be able to beat the classical value of the game, even using remaining non-signaling resources. The strategy $\overline{S}$ being a uniform distribution over strategies where 1 player and its neighbors ignore their resources, we obtain that $$ 9 \varepsilon \geq P(\overline{S} \; loses)  \geq \dfrac{1}{6}$$

From this, we can deduce that $P(S \; loses)\geq \dfrac{1}{54}\simeq 1.85\%.$
\end{proof}

\vspace{1cm}

\noindent{\large  \textbf{Robustness against classical strategies}.}

In this game defined as above, the classical value was $\dfrac{5}{6}$, and it was obtained with 6 questions chosen with equal probability. However, in the MCG associated with the graph $C_5$ \cite{contex}, there are exactly 11 different questions where players can lose. Thus we can try to analyze the game with a distribution over these 11 questions. The goal is to find a symmetric game between the different players for which the classical value is lower than $\dfrac{5}{6}$. We found a game, denoted $\mathcal{G}'_{C_5}$, that had a greater classical value. This is defined by the following set :

\begin{eqnarray*}
  \mathcal{G}'_{C_5} & = & \Big\lbrace \Big(\dfrac{3}{13},Q_a,\lbrace 1,2,3,4,5 \rbrace,1\Big), \Big(\dfrac{1}{13},Q_i,\lbrace i-1,i,i+1 \rbrace,0\Big), \Big(\dfrac{1}{13},Q'_i,\overline{\lbrace i  \rbrace},0\Big)  \Big\rbrace
\end{eqnarray*}

\begin{flushleft}
With $Q_a=1_11_21_31_41_5$, $Q_i=0_{i-2}0_{i-1}1_i0_{i+1}0_{i+2}$ , $Q'_i=0_{i-2}1_{i-1}0_i1_{i+1}0_{i+2}$ and $\overline{\lbrace i  \rbrace}=\lbrace 1,2,3,4,5 \rbrace \setminus \lbrace i \rbrace$.

\vspace{0.5cm}

We proved the following result :
\end{flushleft}
\begin{clm}
The best symmetric game protocol, that has $MCG(C_5)$ as underlying game, to separate quantum against classical players is $\mathcal{G}'_{C_5}$ and it has classical value $\dfrac{10}{13}$.	
\end{clm}
\begin{proof}
A classical strategy is a convex distribution of deterministic behaviour for each player. The four possible behaviours for a player $i$ are : \begin{itemize}
\item $\forall \; x_i \in \lbrace 0,1 \rbrace , a_i=0$ : the strategy is denoted $0$.
\item $\forall \; x_i \in \lbrace 0,1 \rbrace , a_i=1$ : the strategy is denoted $1$
\item $\forall \; x_i \in \lbrace 0,1 \rbrace , a_i=x_i$, denoted $Id$
\item $\forall \; x_i \in \lbrace 0,1 \rbrace , a_i=1+x_i$, denoted $Not$.
\end{itemize}

$\mathcal{G}'_{C_5}$ is symetric for players, thus, two strategies have the same probability to win over the question if they can just be obtained by players permutation from one an other. Thus we can enumerate the different strategies for the players to this game.

The notation will be the following : The global strategy where players 1 and 2 answer $1$, where player 3 and 4 answer $Id$ and player 5 answers $Not$ is denoted by $1_2Id_2Not$ (this is equivalent to $Id_2Not1_2$ for example). The index represents the size of the group of players using this strategy. 

After it, we want to find the best probability distribution over questions to maximize the classical value of the game, moreover, we want to keep the symetry of the game. Thus, questions $Q_i$  will all be asked with the same probability for all $i$, such as questions $Q'_i$. If we call $x$ the weight of questions $Q_a$, $y$ for questions $Q_i$, and $z$ for questions $Q'_i$, then the set defining the game will be the following : \begin{eqnarray*}
  \mathcal{G}'_{C_5} & = & \Big\lbrace \Big(\dfrac{x}{x+5y+5z},Q_a,\lbrace 1,2,3,4,5 \rbrace,1 \Big),
   \Big(\dfrac{y}{x+5y+5z},Q_i,\lbrace i-1,i,i+1 \rbrace,0\Big),\\
   & &  \Big(\dfrac{z}{x+5y+5z},Q'_i,\overline{\lbrace i \rbrace},0\Big)\Big\rbrace
\end{eqnarray*}

We put in \textbf{Appendix} the table showing for each strategy, the weight of questions where players lose.

So, as we try to find the greatest classical value, we want to find values for $x$, $y$ and $z$ that maximizes the minimum for weights of losing questions on all strategies. Thus, we want $x$, $y$ and $z$ such that \textit{min}$\Big(\dfrac{x}{x+5y+5z},\dfrac{y+2z}{x+5y+5z}\Big)$ is maximal. We find $x=3 \; \text{and} \; y=z=1$.
 \end{proof}

\section{Conclusion}

	The gap found here is the best known to separate quantum theories from non-signaling through pseudo-telepathy games. Indeed Reichardt and Chao indicated that, in an unpublished work, they found a gap of $1.2\%$ for $C_5$, which has been improved in this paper. 
	
	In the last part, this robustness against classical strategies could have helped us to find a better gap for this game. But the tools we used did not allow to remove randomness to players as the set of questions was not adapted anymore. 
	
	It would be interesting to extend the method to the study of parallel games, or with product of predicates on graphs.

\bibliography{biblio}{}
\bibliographystyle{alpha} 

\nopagebreak[4]
\section*{Appendix}
\centering
\begin{tabular}{|c|c|}
 
\hline 
\textbf{Classical Strategy} & \textbf{Weight of losing questions} \\ 
\hline 
$0_5$ & \textcolor{red}{\textbf{$x$}} \\ 
\hline 
$1_5$ & $5y+5z$ \\ 
\hline 
 $Id_5$ & \textcolor{red}{\textbf{$x$}} \\ 
\hline 
$Not_5$ & $5y$ \\ 
\hline 
$01_4$ & $x+2y+4z$ \\ 
\hline 
$0Id_4$ & $x+4y+3z$ \\ 
\hline 
$0Not_4$ & $x+2y+2z$ \\ 
\hline 
$10_4$ & $3y+4z$ \\ 
\hline 
$1Id_4$ & $3y+2z$ \\ 
\hline 
$1Not_4$ & \textcolor{red}{\textbf{$y+2z$}} \\ 
\hline 
$Id0_4$ & \textcolor{red}{\textbf{$y+2z$}} \\ 
\hline 
$Id1_4$ & $3y+2z$ \\ 
\hline 
$IdNot_4$ & $y+4z$ \\ 
\hline 
$Not0_4$ & $x+2y+2z$ \\ 
\hline 
$Not1_4$& $x+4y+2z$ \\
\hline
$NotId_4$& $x+2y+4z$ \\
\hline
$0_21_3$ &  $4y+2z$ \\
\hline
$0_2Id_3$ & $3y+4z$ \\
\hline
$0_2Not_3$ & $x+4y+2z$ \\
\hline
$1_20_3$ & $x+2y+2z$ \\
\hline
$1_2Not_3$ & $x+2y+4z$ \\
\hline
$1_2Id_3$ & \textcolor{red}{\textbf{$y+2z$}}\\
\hline
$Not_20_3$ & $x+4y+2z$ \\ 
\hline
$Not_21_3$ & $3y+4z$ \\ 
\hline 
$Not_2Id_3$ & $3y+2z$ \\ 
\hline 
$Id_20_3$ & $x+2y+4z$ \\ 
\hline 
$Id_21_3$ & \textcolor{red}{\textbf{$y+2z$}} \\ 
\hline 
$Id_2Not_3$ & $x+2y+2z$ \\ 
\hline 
$01Id_3$ & $x+4y+2z$ \\ 
\hline 
$01Not_3$ & \textcolor{red}{\textbf{$y+2z$}} \\ 
\hline 
$0Id1_3$ & $x+2y+2z$ \\ 
\hline 
$0IdNot_3$ & $3y+2z$ \\ 
\hline 
$0Not1_3$ & $3y+4z$ \\ 
\hline 
$0NotId_3$ & $3y+4z$ \\ 
\hline 
$1Id0_3$ & $x+2y+4z$ \\ 
\hline 
$1IdNot_3$ & $x+2y+4z$ \\ 
\hline 
$1Not0_3$ & $3y+2z$ \\ 
\hline 
$1NotId_3$ & $3y+4z$ \\ 
\hline

\end{tabular}

\begin{table}
	\begin{tabular}{|c|c|}
\hline
\textbf{Classical Strategy} & \textbf{Weight of losing questions} \\ 
\hline

$NotId0_3$ & \textcolor{red}{\textbf{$y+2z$}} \\ 
\hline 
$NotId1_3$ & $x+4y+2z$ \\ 
\hline 
$0_21_2Not$ & $x+4y+2z$ \\ 
\hline 
$0_21_2Id$ & \textcolor{red}{\textbf{$y+2z$}} \\ 
\hline 
$0_2Not_21$ & $3y+4z$ \\ 
\hline 
$0_2Not_2Id$ & $3y+4z$ \\ 
\hline 
$0_2Id_21$ & $3y+2z$ \\ 
\hline 
$0_2Id_2Not$ & $x+2y+2z$ \\ 
\hline 
$1_2Id_2Not$ & $x+2y+4z$ \\ 
\hline 
$1_2Id_20$ & $x+2y+4z$ \\ 
\hline 
$1_2Not_20$ & $x+2y2z$ \\ 
\hline 
$1_2Not_2Id$ & $3y+2z$ \\ 
\hline 
$Not_2Id_21$ & $x+y+2z$ \\ 
\hline 
$Not_2Id_20$ & $x+4y+2z$ \\ 
\hline 
$Id_201Not$ & $x+y+2z$ \\ 
\hline 
$0_21IdNot$ & $x+4y+2z$ \\ 
\hline 
$1_20IdNot$ & $3y+2z$ \\ 
\hline 
$Not_201Id$ & $x+4y+2z$ \\ 
\hline 
	\end{tabular}
	\centering
	\caption{Weight of losing questions for Strategies}
\end{table}
\end{document}